\documentclass[11pt]{article} 
\usepackage{fullpage}
\usepackage{amsmath}
\usepackage{amssymb}
\usepackage{booktabs}
\usepackage{multirow}
\usepackage{footmisc}
\usepackage{xspace}
\usepackage[section]{placeins}
\usepackage[algosection]{algorithm2e}
\usepackage{amsfonts}
\usepackage{enumerate}
\usepackage{amsthm}

\newtheorem{theorem}{Theorem}[section]
\newtheorem{lemma}[theorem]{Lemma}

\newtheorem{corollary}[theorem]{Corollary}

\numberwithin{table}{section}
\numberwithin{figure}{section}
\numberwithin{equation}{section}

\title{Autoreducibility of NP-Complete Sets\footnote{This research was support in part by NSF grant
    0917417.}}
\author{
  John M. Hitchcock
  and
  Hadi Shafei\\
  Department of Computer Science\\
  University of Wyoming
}
\date{}

\newcommand{\INC}{}

\newenvironment{lemma_cite}[1]
{\begin{lemma}  {\rm (#1)}}
{\end{lemma}}

\newcommand{\cupdot}{\dot\cup}
\newcommand{\QUERY}{\mathrm{QUERY}}
\newcommand{\pgenericity}{$\p$-genericity\xspace}
\newcommand{\pgeneric}{$\p$-generic\xspace}

\newcommand{\dtt}{\mathrm{dtt}}
\newcommand{\pair}[1]{\langle #1 \rangle}
\newcommand{\N}{\mathbb{N}}
\newcommand{\NP}{{\ensuremath{\mathrm{NP}}}} 

\newcommand{\calR}{{\cal R}}
\newcommand{\calC}{{\cal C}}
\newcommand{\EXP}{{\rm EXP}}
\newcommand{\NEXP}{{\rm NEXP}}
\newcommand{\p}{{\mathrm{p}}}
\renewcommand{\P}{\ensuremath{{\mathrm P}}}
\newcommand{\coNP}{\co{\NP}}
\newcommand{\co}[1]{\mathrm{co}#1}
\newcommand{\SAT}{{\rm SAT}}
\newcommand{\DTIME}{\mathrm{DTIME}}
\newcommand{\m}{\mathrm{m}}
\newcommand{\btt}{\mathrm{btt}}
\renewcommand{\tt}{\mathrm{tt}}
\newcommand{\T}{\mathrm{T}}
\newcommand{\leqpm}{\leqp_\mathrm{m}}
\newcommand{\leqp}{\leq^\p}

\begin{document}
\maketitle

\begin{abstract}

  We study the polynomial-time autoreducibility of $\NP$-complete sets
  and obtain separations under strong hypotheses for $\NP$. Assuming
  there is a \pgeneric set in $\NP$, we show the following:
  \begin{itemize}
    \item For every $k \geq 2$, there is a $k$-T-complete set for
      $\NP$ that is
      $k$-T autoreducible, but is not $k$-tt autoreducible or $(k-1)$-T
      autoreducible.
      \item For every $k \geq 3$, there is a $k$-tt-complete set for
        $\NP$ that
        is $k$-tt autoreducible, but is not $(k-1)$-tt autoreducible
        or $(k-2)$-T autoreducible.
    \item There is a tt-complete set for $\NP$ that is tt-autoreducible,
        but is not btt-autoreducible.
  \end{itemize}
  
\noindent
  Under the stronger assumption that there is a \pgeneric set in
  $\NP\cap\coNP$, we show:
  \begin{itemize}
    \item For every $k \geq 2$, there is a $k$-tt-complete set for
      $\NP$ that is $k$-tt autoreducible, but is not $(k-1)$-T
      autoreducible.
  \end{itemize}
  
\noindent
Our proofs are based on constructions from separating NP-completeness
notions. For example, the construction of a 2-T-complete set for
$\NP$ that is not 2-tt-complete also separates 2-T-autoreducibility
from 2-tt-autoreducibility.

\end{abstract}

\section{Introduction}

Autoreducibility measures the redundancy of a set. 
For a reducibility $\calR$, a set $A$ is $\calR$-autoreducible if
there is a $\calR$-reduction from $A$ to $A$ where the instance is
never queried \cite{Tra70}.  Understanding the autoreducibility of
complete sets is important because of applications to separating
complexity classes \cite{BFvMT00}. We study the polynomial-time
autoreducibility \cite{Amb83} of $\NP$-complete sets.

Natural problems are paddable and easily shown to be
m-autoreducible. In fact, Gla{\ss}er et
al. \cite{Glasser07} showed that all nontrivial m-complete sets for $\NP$
and many other complexity classes are m-autoreducible.  Beigel and
Feigenbaum \cite{BeiFei92} showed that T-complete sets for $\NP$ and
the levels of the polynomial-time hierarchy are T-autoreducible.
We focus on intermediate reducibilities between many-one and Turing.

Previous work has studied separations of these autoreducibility
notions for larger complexity classes. Buhrman et al. \cite{BFvMT00}
showed there is a 3-tt-complete set for $\EXP$ that is not
btt-autoreducible. For $\NEXP$, Nguyen and Selman
\cite{NguyenSelman:STACS14} showed there is a 2-T-complete set that is
not 2-tt-autoreducible and a tt-complete set that is not
btt-autoreducible.  We investigate whether similar separations hold
for $\NP$.

Since all $\NP$ sets are 1-tt-autoreducible if $\P = \NP$, it
is necessary to use a hypothesis at least as strong as $\P \neq \NP$
to separate autoreducibility notions.  We work with the {\em
  Genericity Hypothesis} that there is a \pgeneric set in $\NP$
\cite{AmFlHu87,AmbBen00}. This is stronger than $\P \neq \NP$, but weaker
than the {\em Measure Hypothesis} \cite{Lutz:CvKL,Hitchcock:HHDCC}
that there is a p-random set in $\NP$. Under the Genericity
Hypothesis, we separate many autoreducibility notions for
$\NP$-complete sets. Our main results are summarized in Table \ref{table:results}.

Previous work has used the measure and genericity hypotheses to
separate completeness notions for $\NP$. Consider the set $$C = G
\dot\cup (G \cap \SAT) \dot\cup (G \cup \SAT),$$ where $G \in \NP$ and
$\cupdot$ is disjoint union. Then $C$ is 2-T-complete for $\NP$, and
if $G$ is \pgeneric, $C$ is not 2-tt-complete
\cite{Lutz:CvKL,AmbBen00}. There is a straightforward 3-T (also 5-tt)
autoreduction of $C$ based on padding $\SAT$.\footnote{Given an
  instance $x$ of $C$, pad $x$ to an instance $y$ such that
  $\SAT[x]=\SAT[y]$. We query $G[y]$ and then query either $G \cap
  \SAT[y]$ if $G[y] = 1$ or $G \cup \SAT[y]$ if $G[y]=0$ to learn
  $\SAT[y]$. Finally, if our instance is $G[x]$ the answer is obtained
  by querying $G\cap\SAT[x]$ if $\SAT[y]=1$ or by querying
  $G\cup\SAT[x]$ if $\SAT[y]=0$. If our instance is $G\cup\SAT[x]$ or
  $G\cap\SAT[x]$, we query $G[x]$ and combine that answer with
  $\SAT[y]$.}  However, since $C$ is 2-T-honest-complete, we
indirectly obtain a 2-T (also 3-tt) autoreduction by first reducing
through $\SAT$ (Lemma \ref{le:honest auto}). In Theorem \ref{th:2T vs
  2tt} we show $C$ is not 2tt-autoreducible.

It turns out this idea works in general. We show that many sets which
separate completeness notions also separate autoreducibility
notions. Ambos-Spies and Bentzien \cite{AmbBen00} also separated both
$k$-T-completeness and $(k+1)$-tt-completeness from both
$k$-tt-completeness and $(k-1)$-T-completeness for every $k \geq 3$
under the Genericity Hypothesis. We show that the same sets also
separate $k$-T-autoreducibility and $(k+1)$-tt-autoreducibility from
$k$-tt-autoreducibility and $(k-1)$-T-autoreducibility (Theorems
\ref{th:k-tt} and \ref{th:k-1-T}).  We also obtain that there is a
tt-complete set for $\NP$ that is tt-autoreducible and not
btt-autoreducible (Theorem \ref{th:btt vs tt}), again using a
construction of Ambos-Spies and Bentzien.

In the aforementioned results, there is a gap -- we only separate
$k$-tt-autoreducibility from $(k-2)$-T-autoreducibility (for $k \geq
3$), where we can hope for a separation from
$(k-1)$-T-autoreducibility. The separation of $k$-tt from $(k-1)$-T is
also open for completeness under the Genericity Hypothesis (or the
Measure Hypothesis). To address this gap, we use a stronger hypothesis
on the class $\NP \cap \coNP$.  Pavan and Selman \cite{PavSel04}
showed that if $\NP \cap \coNP$ contains a
$\DTIME(2^{n^\epsilon})$-bi-immune set, then 2-tt-completeness is
different from 1-tt-completeness for $\NP$. We show that if $\NP \cap
\coNP$ contains a \pgeneric set, then $k$-tt-completeness is different
from $(k-1)$-T-completeness for all $k \geq 3$ (Theorem \ref{th:k-tt
  vs (k-1)-T complete}). We then show these constructions also
separate autoreducibility: if there is a \pgeneric set in
$\NP\cap\coNP$, then for every $k \geq 2$, there is a $k$-tt-complete
set for $\NP$ that is $k$-tt autoreducible, but is not $(k-1)$-T
autoreducible (Theorems \ref{th:2tt vs 1tt} and \ref{th:k-tt vs (k-1)-T autored}).

This paper is organized as follows. Preliminaries are in Section
\ref{sec:prelim}. The results using the Genericity Hypothesis are
presented in Section \ref{sec:genericity}. We use the stronger
hypothesis on $\NP\cap\coNP$ in Section \ref{sec:stronger}. Section
\ref{sec:conclusion} concludes with some open problems.

\newcommand{\calS}{\mathcal{S}}

\begin{table}[ht!]
\begin{center}
\begin{tabular}{|c|c|c|c|}
\hline
$\calC$ & $\calS$ & $\calR$ & notes \\
\hline
\hline
$\NP$ & $k$-T & $k$-tt & 
Theorem \ref{th:2T vs 2tt} ($k=2$), Theorem
\ref{th:k-tt} ($k \geq 3)$\\
\hline
$\NP$ & $k$-T & $(k-1)$-T & 
Theorem \ref{th:2T vs 2tt} ($k=2$), Theorem
\ref{th:k-1-T} ($k \geq 3)$\\
\hline
$\NP$ & $k$-tt & $(k-1)$-tt & Corollary \ref{co:3tt vs 2tt} ($k=3$), Theorem \ref{th:k-tt} ($k \geq 4$)\\
\hline
$\NP$ & $k$-tt & $(k-2)$-T & Corollary \ref{co:3tt vs 1T} ($k=3$), Theorem \ref{th:k-1-T} ($k \geq 4$)\\
\hline
$\NP$ & tt & btt & Theorem \ref{th:btt vs tt}\\
\hline
$\NP \cap \coNP$ & $k$-tt & $(k-1)$-T & Theorem \ref{th:2tt vs 1tt}
($k = 2$), Theorem \ref{th:k-tt vs (k-1)-T autored} ($k \geq 3$) \\
\hline
\end{tabular}
\caption{If $\calC$ contains a \pgeneric set, then there is a $\calS$-complete
set in $\NP$ that is $\calS$-autoreducible but not $\calR$-autoreducible.}
\label{table:results}\end{center}
\end{table}

\section{Preliminaries}\label{sec:prelim}

We use the standard enumeration of binary strings, i.e $s_0 = \lambda,
s_1 = 0, s_2 = 1,s_3 = 00,...$ as an order on binary strings. All
languages in this paper are subsets of $\{0,1\}^*$ identified with
their characteristic sequences. In other words, every language $A \in
\{0,1\}^*$ is identified with $\chi_A = A[s_0]A[s_1]A[s_2]...$. If $X$
is a set, equivalently a binary sequence, and $x \in \{0,1\}^*$ then
$X\upharpoonright x$ is the initial segment of $X$ for all strings
before $x$, i.e the subset of $X$ that contains every $y \in X$ that
$y < x$.

All reductions in this paper are polynomial-time reductions, therefore
we may not emphasize this every time we define a reduction. 
We use standard notions of reducibilities \cite{LaLySe75}.

Given $A$, $B$, and $\calR \in
\{\m,\;\T,\;\tt,\;k\text{-}\T,\;k\text{-}\tt,\;\btt\}$, $A$ is {\em
  polynomial-time $\calR$-honest reducible} to $B$ ($A
\le_{\calR\text{-}h}^\p$) if $A \le_\calR^\p$ and there exist a
constant $c$ such that for every input $x$, every query $q$ asked from
$B$ has the property $|x|^{1/c} < |q|$. In particular, a reduction
$\calR$ is called {\em length-increasing} if on every input the
queries asked from the oracle are all longer than the input.

For any reduction $\calR \in \{\m,\;\T,\;\tt,\;k\text{-}\T,\;k\text{-}\tt,\;\btt\}$ a 
language $A$ is $\calR$-{\em autoreducible} if $A \leq^\p_\calR$ via a
reduction where on every instance $x$, $x$ is not queried.

The following lemma states that any honest-complete set for $\NP$ is
also autoreducible under the same type of reduction. This follows
because $\NP$ has a paddable, length-increasing complete set.

\begin{lemma}
\label{le:honest auto}
Let $\calR \in \{\m,\;\T,\;\tt,\;k\text{-}\T,\;k\text{-}\tt,\;\btt,\ \ldots\}$
be a reducibility. Then every $\calR$-honest-complete set for $\NP$
 is $\calR$-autoreducible.
\end{lemma}
\begin{proof}
Let $A \in \NP$ be $\calR$-honest-complete. Then there is an $\calR$-honest
reduction $M$ from $\SAT$ to $A$. There exists $m \geq 1$ such that every
query $q$ output by $M$ on an instance $x$ satisfies $|q| \geq
|x|^{\tfrac{1}{m}}$.

Since $\SAT$ is $\NP$-complete via length-increasing many-one
reductions, $A \leqpm \SAT$ via a length-increasing reduction $g$.
Since $\SAT$ is paddable, there is a polynomial-time function $h$ such
that for any $y$, $\SAT[h(y)] = \SAT[y]$ and $|h(y)| > |y|^m$.

To obtain our $\calR$-autoreduction of $A$, we combine $g$, $h$, and $M$. 
On instance $x$ of $A$, compute the instance $h(g(x))$ of $\SAT$ and use $M$ to reduce $h(g(x))$ to
$A$. Since $|h(g(x))| > |g(x)|^m > |x|^m$, every query $q$ of $M$
has $|q| > |h(g(x))|^{\tfrac{1}{m}} > |x|$. Therefore all queries are different
than $x$ and this is an autoreduction.
\end{proof}

Most of the results in this paper are based on a non-smallness
hypothesis for $\NP$ called the {\em Genericity Hypothesis} that $\NP$
contains a $\p$-generic set \cite{AmFlHu87,AmbBen00}. In order to define genericity first we need to define what a 
{\em simple extension function} is.
 For any $k$, a simple $n^k$-extension function is a partial function from 
$\{0,1\}^*$ to $\{0,1\}$ that is computable in $O(n^k)$. Given a set $A$ and an extension function $f$ 
we say that $f$ is {\em dense along} $A$ if $f$ is defined on infinitely many initial segments of $A$. 
A set $A$ {\em meets} a simple extension function $f$ at $x$ if $f (A \upharpoonright x)$ is defined and
equal to $A[x]$. We say $A$ meets $f$ if $A$ meets $f$ at some $x$.
A set $G$ is called {\em \pgeneric} if it meets every simple $n^k$-extension function for any $k \geq 1$ \cite{AmbBen00}. 
A partial function $f:\{0,1\}^* \rightarrow (\{0,1\}^* \times
\{0,1\})^*$ is called a {\em k-bounded extension function} if whenever $f(X \upharpoonright x)$ is defined, 
$f(X \upharpoonright x) = (y_0,i_0)...(y_m,i_m)$ for some $m < k$, and $x \le y_0 < y_1<...<y_m$, 
where $y_j$'s are strings and $i_j$'s are either $0$ or $1$. A set $A$ meets $f$ at $x$ if 
$f(A \upharpoonright x)$ is defined, and $A$ agrees with $f$ on all  $y_j$'s, i.e. 
if $f(A \upharpoonright x) = (y_0,i_0)...(y_m,i_m)$ then $A[y_j] = i_j$ for all $j \le m$ \cite{AmbBen00}.

We will use the following routine extension of a lemma in \cite{AmbBen00}.

\begin{lemma}\label{extension function revised}
Let $l,c \geq 1$ and let $f$ be an $l$-bounded partial extension function defined on initial segments 
$\alpha = X \upharpoonright 0^n$ of length $2^n$ $(n \geq 1)$. Whenever $f(\alpha)$ is defined we have
\[ 
f(\alpha) = (y_{\alpha , 1} , i_{\alpha , 1}),...,(y_{\alpha , l_{\alpha}} , i_{\alpha , l{\alpha}}),
\]
where $l_{\alpha} \leq l$, $pos(\alpha) = (y_{\alpha , 1},...,y_{\alpha,l_{\alpha}})$ is computable in 
$2^{cn}$ steps and $i_{\alpha , j}$ is computable in $2^{c|y_{\alpha , j}|}$ steps. 
Then for every \pgeneric set $G$, if $f$ is dense along $G$ then $G$ meets $f$.
\end{lemma}

\section{Autoreducibility Under the Genericity Hypothesis}\label{sec:genericity}

We begin by showing the Genericity Hypothesis implies there is a
$2$-$\T$-complete set that separates $2$-$\T$-autoreducibility from
$2$-$\tt$-autoreducibility. The proof utilizes the construction of
\cite{Lutz:CvKL,AmbBen00} that of a set that separates
$2$-$\T$-completeness from $2$-$\tt$-completeness.

\begin{theorem}
\label{th:2T vs 2tt}
If $\NP$ contains a \pgeneric language, then there exists a $2$-$\T$-complete set in $\NP$ that is $2$-$\T$-autoreducible, but not $2$-$\tt$-autoreducible.
\end{theorem}

\begin{proof}
Let $G \in \NP$ be \pgeneric and define $C = G \;\dot\cup \;(G \cap \SAT)
\;\dot\cup \;(G \cup \SAT)$, where $\dot\cup$ stands for disjoint union
\cite{Lutz:CvKL,AmbBen00}. Disjoint union can be implemented by adding
a unique prefix to each set and taking their union. To be more clear,
let $C = 0G \;\cup\;10(G\cap \SAT)\;\cup\;11(G\cup \SAT)$. It follows
from closure properties of $\NP$ that $C \in \NP$.

To see that $C$ is $2$-$\T$-complete, consider an oracle Turing machine
$M$ that on input $x$ first queries $0x$ from $C$. If the answer is
positive, i.e. $x \in G$, $M$ queries $10x$ from $C$, and outputs the
result. Otherwise, $M$ queries $11x$ from $C$, and outputs the
answer. This Turing machine always makes two queries from $C$, runs in
polynomial time, and $M^C(x) = \SAT[x]$. This completes the proof that
$C$ is also $2$-$\T$-completeness.  Since all queries from $\SAT$ to $C$
are length-increasing, it follows from Lemma \ref{le:honest auto} that
$C$ is $2$-$\T$-autoreducible.

The more involved part of the proof is to show that $C$ is not $2$-$\tt$-autoreducible. To get a
contradiction assume that $C$ is $2$-$\tt$-autoreducible. This means there exist polynomial-time 
computable functions $h$, $g_1$, and $g_2$ such that for every $x \in \{0, 1\}^*$,
\[ C[x] = h(x , C[g_1(x)], C[g_2(x)]) \] 
and moreover $g_i(x) \neq x$ for $i = 1, 2$. Note that W.L.O.G. we can assume that $g_1(x) < g_2(x)$. \\
For $x = 0z$, $10z$, or $11z$ define the value of $x$ to be $z$, and let $x = 0z$ for some string $z$.
We have:
\[ C[x]= G[z] = h(x , C[g_1(x)], C[g_2(x)]) \]
To get a contradiction, we consider different cases depending on
whether some of  the queries have the same value as $x$ or not, and
the Boolean function $h(x,.,.)$. For some of these cases we show they
can happen only for finitely many $z$'s, and for the rest we show that
$\SAT[z]$ can be decided in polynomial time.
As a result $\SAT$ is decidable in polynomial time a.e., which contradicts the assumption that $\NP$
contains a \pgeneric language.

  \begin{itemize}
\item The first case is when values of $g_1(x)$ and $g_2(x)$ are different from $z$, and also different from each other. Assume this happens for infinitely many $z$'s. We define an extension function $f$
that is dense along $G$, so $G$ has to meet it, but $f$ is defined in a way that if $G$ meets $f$, the autoreduction will be refuted. In order to define the value that $f$ forces to $G[z]$ on the right
hand side of the reduction, we define a function $\alpha$ that assigns $0$ or $1$ to queries of our autoreduction. The idea behind defining $\alpha$ is that its value on queries $q_i$ is equal to $C[q_i]$
after we forced appropriate values into $G$, but computation of
$\alpha$ can be done in at most $2^{2n}$ steps (given access to the
partial characteristic sequence of $G$).\\
\[ \alpha(w) = \begin{cases} 
      C[w] & \textrm{ if $w < x$} \\
      0 & \textrm{ if $w > x$ and $w$ = 0y or 10y for some y} \\
      1 & \textrm{ if $w > x$ and $w$ = 11y for some y}\\  
   \end{cases} \]
Note that in the first case, since $w < x$, the value of $C[w]$ is computable in $2^{2n}$ steps. Let 
$j = h(x , \alpha (g_1(x)) , \alpha (g_2(x)))$. Later, when defining the extension function, we force 
the value of $C[x] = G[z]$ to be $1-j$, hence refuting the autoreduction.\\
The extension function $f$ is defined whenever this case happens, and it forces three values into $G$.
If $g_i(x) = 0v$ or $10v$ for some $v$, then $f(x)$ forces $G[v] = 0$. If $g_i(x) = 11v$ for some string $v$ then $f(x)$ forces $G[v] = 1$. Finally, $f(x)$ forces $G[z] = 1 - j$. Since we assumed that this case happens for infinitely many $x$'s, $f$ is dense along $G$. Therefore $G$ must meet $f$ at some string $x = 0z$. But by the very definition of $f$ this refutes the autoreduction. Hence this case can happen only for finitely many $x$'s.

\item In this case we consider the situation that $g_1(x)$ and $g_2(x)$ have the same value, say $y$, but $y \neq z$. If $y < z$ we can compute $C[g_1(x)]$ and $C[g_2(x)]$ and force 
$G[z] = 1 - h(x, C[g_1(x)], C[g_2(x)]$, which refutes the autoreduction. Therefore this cannot happen i.o.
Now based on the prefixes of $g_1(x)$ and $g_2(x)$ we consider the following cases:
\begin{enumerate}
\item If $g_1(x) = 0v$ and $g_2(x) = 10v$ we force $G[v] = 0$ and $G[x] = 1 - h(x, 0 , 0)$. This refutes the autoreduction, therefore this case can happen only finitely many times.
\item If $g_1(x) = 0v$ and $g_2(x) = 11v$ we force $G[v] = 1$ and $G[x] = 1 - h(x, 1 , 1)$. This also refutes the autoreduction, so it cannot happen i.o.
\end{enumerate}
The only possibility that remains in this case is $g_1(x) = 10v$ and $g_2(x) = 11v$. In this case the autoreduction equality can be stated as:
\[G[z] = h(x , G \cap \SAT[v] , G \cup \SAT[v]) \]
To show that this also cannot happen i.o. we need to look into different cases of the Boolean function
$h(x,.,.)$.

\begin{enumerate}

\item If $h(x , a, b) = 0$, or $1$, then force $G[z] = 1$ or $0$ respectively. Therefore this Boolean function can occur only finitely many times.

\item If $h(x, a, b) = a$, in other words $G[z] = G \cap \SAT[v]$, force $G[z] = 1$ and $G[v] = 0$. This refutes the autoreduction,  so this Boolean function cannot happen i.o.

\item If $h(x, a, b) = \neg a$, in other words $G[z] = \neg G \cap \SAT[v]$, force $G[z] = 0$ and $G[v] = 0$. This refutes the autoreduction,  so this Boolean function cannot happen i.o.

\item If $h(x, a, b) = b$, in other words $G[z] = G \cup \SAT[v]$, force $G[z] = 0$ and $G[v] = 1$. This refutes the autoreduction,  so this Boolean function cannot happen i.o.

\item If $h(x, a, b) = \neg b$, in other words $G[z] = \neg G \cup \SAT[v]$, force $G[z] = 1$ and $G[v] = 1$. This refutes the autoreduction,  so this Boolean function cannot happen i.o.

\item If $h(x, a, b) = a \wedge b$, in other words $G[z] = (G \cap \SAT[v]) \wedge (G \cup \SAT[v])$, but this is equal to $G \cap \SAT[v]$. Therefore this case is similar to the second case.

\item If $h(x, a, b) =\neg a \wedge b$, in other words $G[z] = \neg(G \cap \SAT[v]) \wedge (G \cup \SAT[v])$. Force $G[z] = 1$ and $G[v] = \SAT[v]$. This contradicts the autoreduction equality. Therefore this case can happen only finitely many times.

\item If $h(x, a, b) = a \wedge \neg b$, in other words $G[z] = (G \cap \SAT[v]) \wedge \neg(G \cup \SAT[v])$, forcing $G[z] = 1$ refutes the autoreduction.

\item If $h(x, a, b) = \neg a \wedge \neg b$, in other words $G[z] = \neg(G \cap \SAT[v]) \wedge \neg (G \cup \SAT[v])$, but this is equal to $\neg G \cup \SAT[v]$. Therefore this case is similar to the fifth case.

\item If $h(x, a, b) = a \vee b$, in other words $G[z] = (G \cap \SAT[v]) \vee (G \cup \SAT[v])$, but this is equal to $G \cup \SAT[v]$. Therefore this case is similar to the fourth case.

\item If $h(x, a, b) = \neg a \vee b$, in other words $G[z] = \neg(G \cap \SAT[v]) \vee (G \cup \SAT[v])$. In this case forcing $G[z] = 0$ refutes the autoreduction.

\item If $h(x, a, b) = a \vee \neg b$, in other words $G[z] = (G \cap \SAT[v]) \vee \neg(G \cup \SAT[v])$. In this case forcing $G[z] = 0$ and $G[v] = \SAT[v]$ refutes the autoreduction.

\item If $h(x, a, b) = \neg a \vee \neg b$, in other words $G[z] = \neg(G \cap \SAT[v]) \vee \neg(G \cup \SAT[v])$, but this is equal to $\neg(G \cap \SAT[v])$. Therefore this case is similar to the third case.

\item If $h(x, a, b) = a \leftrightarrow b$, in other words $G[z] = (G \cap \SAT[v]) \leftrightarrow (G \cup \SAT[v])$. In this case $G[z] = 0$ and $G[v] = \SAT[v]$ refutes the autoreduction.

\item If $h(x, a, b) = \neg a \leftrightarrow b$, in other words $G[z] = \neg (G \cap \SAT[v]) \leftrightarrow (G \cup \SAT[v])$. In this case $G[z] = 1$ and $G[v] = \SAT[v]$ refutes the autoreduction.

\end{enumerate}
We exhaustively went through all possible Boolean functions for the case where both queries have the same value which is different from the value of $x$, and showed that each one of them can happen only for finitely many $x$'s. As a result this case can happen only for finitely many $x$'s.

\item This is the case when one of the queries, say $g_1(x)$ has the same value as $x$, but the other query has a different value. We only consider the case where $g_1(x) = 10z$. The other case, i.e. $g_1(x) = 11z$ can be done in a similar way. Again, we need to look at different possibilities for the Boolean function $h(x,.,.)$.

\begin{enumerate}

\item $h(x, a, b) = 0$ or $1$. Forcing $G[z] = 1$ or $0$ respectively refutes the autoreduction.

\item $h(x, a, b) = a$, i.e. $G[z] = G \cap \SAT [z]$. If this happens i.o with $\SAT[z] = 0$ then we can refute the autoreduction by forcing $G[z] = 0$. Therefore in this case $\SAT[z] = 1$ a.e.

\item $h(x, a, b) = \neg a$, i.e. $G[z] = \neg (G \cap \SAT [z])$. By forcing $G[z] = 0$ we can refute the reduction. Therefore this case cannot happen i.o.

\item $h(x, a, b) = b$ or $\neg b$. Similar to previous cases.

\item $h(x, a, b) = a \wedge b$, i.e. $G[z] = (G \cap \SAT[z]) \wedge C[g_2(x)]$. In this case $\SAT[z]$ has to be $1$ a.e.

\item $h(x, a, b) =\neg a \wedge b$, i.e. $G[z] = \neg (G \cap \SAT[z]) \wedge C[g_2(x)]$. If $g_2(x) = 0y$ or $10y$ for some $y$, then forcing $G[z] = 1$ and $G[v] = 0$ refutes the reduction. If $g_2(x) = 11y$ then we have $G[z] = \neg (G \cap \SAT[z]) \wedge (G \cup \SAT[y]$. Here we force $G[z] = 0$ and $G[y] = 1$.

\item $h(x, a, b) = a \wedge \neg b$, i.e. $G[z] = (G \cap \SAT[z]) \wedge \neg C[g_2(x)]$. In this case 
$\SAT[z] = 1$ a.e.

\item $h(x, a, b) =\neg a \wedge \neg b$, i.e. $G[z] = \neg (G \cap \SAT[z]) \wedge \neg C[g_2(x)]$. If $g_2(x) = 0y$ or $11y$ for some $y$, then forcing $G[z] = 1$ and $G[v] = 1$ refutes the reduction. If $g_2(x) = 10y$ then we have $G[z] = \neg (G \cap \SAT[z]) \wedge \neg (G \cap \SAT[y])$. Here we force $G[z] = 0$ and $G[y] = 0$.

\item $h(x, a, b) = a \vee b$, i.e. $G[z] = (G \cap \SAT[z]) \vee C[g_2(x)]$. If $g_2(x) = 0y$ or $11y$ for some $y$, then forcing $G[z] = 0$ and $G[v] = 1$ refutes the reduction. If $g_2(x) = 10y$ then we have $G[z] = (G \cap \SAT[z]) \vee (G \cap \SAT[y]$. This implies that $\SAT[z]$ must be $1$ a.e.

\item $h(x, a, b) =\neg a \vee b$, i.e. $G[z] = \neg (G \cap \SAT[z]) \vee C[g_2(x)]$. In this case forcing $G[z] = 0$ refutes the reduction.

\item $h(x, a, b) = a \vee \neg b$, i.e. $G[z] = (G \cap \SAT[z]) \vee \neg C[g_2(x)]$. If $g_2(x) = 0y$ or $10y$ for some $y$, then forcing $G[z] = 0$ and $G[v] = 0$ refutes the reduction. If $g_2(x) = 11y$ then we have $G[z] = (G \cap \SAT[z]) \vee \neg(G \cup \SAT[y]$. This implies that $\SAT[z]$ must be $1$ a.e.

\item $h(x, a, b) =\neg a \vee \neg b$, i.e. $G[z] = \neg (G \cap \SAT[z]) \vee \neg C[g_2(x)]$. In this case forcing $G[z] = 0$ refutes the reduction.

\item $h(x, a, b) = a \leftrightarrow b$, i.e. $G[z] = (G \cap \SAT[z]) \leftrightarrow C[g_2(x)]$.
If $g_2(x) = 0y$ or $10y$ for some string $y$, then by forcing $G[z] = 0$ and $G[y] = 0$ we can refute the autoreduction. If $g_2(x) = 11y$, then we have $G[z] = (G \cap \SAT[z]) \leftrightarrow (G \cup \SAT[y])$. This implies that $\SAT[z] = 1$ a.e.

\item $h(x, a, b) = \neg a \leftrightarrow b$, i.e. $G[z] = \neg (G \cap \SAT[z]) \leftrightarrow C[g_2(x)]$. If $g_2(x) = 0y$ or $11y$ for some string $y$, then by forcing $G[z] = 0$ and $G[y] = 1$ we can refute the autoreduction. If $g_2(x) = 10y$, then we have $G[z] = \neg (G \cap \SAT[z]) \leftrightarrow (G \cup \SAT[y])$. This implies that $\SAT[z] = 1$ a.e.

\end{enumerate}

\item In this case we consider the situation where both queries $g_1(x)$ and $g_2(x)$ have the same value as $x$. In other words, in this case we have $g_1(x) = 10z$ and $g_2(x) = 11z$. Therefore we have:
\[G[z] = h(x, G \cap \SAT[z], G \cup \SAT[z]) \]
To investigate this case we need to look at different Boolean functions for $h(x,.,.)$.

\begin{enumerate}

\item $h(x, a, b) = 0$, $1$, $a$, $\neg a$, $b$, or $\neg b$. Each of these cases is similar to one of the cases discussed previously.

\item $h(x, a, b) = a \wedge b$, i.e. $G[z] = G \cap \SAT[z]$. This is also similar to one of the cases that we discussed previously.

\item $h(x, a, b) = \neg a \wedge b$, i.e. $G[z] = \neg (G \cap \SAT[z]) \wedge (G \cup \SAT[z])$. In this case $\SAT[z]$ must be $0$ a.e.

\item $h(x, a, b) = a \wedge \neg b$, i.e. $G[z] = (G \cap \SAT[z]) \wedge \neg(G \cup \SAT[z])$. Forcing $G[z] = 1$ refutes the reduction.

\item $h(x, a, b) = \neg a \wedge \neg b$, i.e. $G[z] = \neg (G \cap \SAT[z]) \wedge \neg(G \cup \SAT[z])$.
This is equal to $\neg(G \cup \SAT[z])$. Therefore forcing $G[z] = 0$ refutes the reduction.

\item $h(x, a, b) = a \vee b$, i.e. $G[z] = (G \cap \SAT[z]) \vee (G \cup \SAT[z])$, which is equal to 
$G \cup \SAT[z]$. Therefore $\SAT[z]$ must be $0$ a.e.

\item $h(x, a, b) = \neg a \vee b$, i.e. $G[z] = \neg (G \cap \SAT[z]) \vee (G \cup \SAT[z])$. In this case $\SAT[z]$ must be $0$ a.e.

\item $h(x, a, b) = a \vee \neg b$, i.e. $G[z] = (G \cap \SAT[z]) \vee \neg(G \cup \SAT[z])$. This implies that $\SAT[z]$ must be $1$ a.e.

\item $h(x, a, b) = \neg a \vee \neg b$, i.e. $G[z] = \neg (G \cap \SAT[z]) \vee \neg (G \cup \SAT[z])$, which is equal to $\neg (G \cap \SAT[z])$. Therefore forcing $G[z] = 0$ refutes the autoreduction.

\item $h(x, a, b) = a \leftrightarrow b$, i.e.  $G[z] = (G \cap \SAT[z]) \leftrightarrow (G \cup \SAT[z])$.
In this case $\SAT[z]$ has to be $1$ a.e.

\item $h(x, a, b) = \neg a \leftrightarrow b$, i.e.  $G[z] = \neg (G \cap \SAT[z]) \leftrightarrow (G \cup \SAT[z])$. This implies that $\SAT[z]$ has to be $0$ a.e.
\end{enumerate}
\end{itemize}

\end{proof}

\begin{corollary}\label{co:3tt vs 2tt}
If $\NP$ contains a \pgeneric language, then there exists a
$3$-$\tt$-complete set for $\NP$ that is $3$-$\tt$-autoreducible, but not $2$-$\tt$-autoreducible. 
\end{corollary}
\begin{proof}
This follows immediately from Theorem \ref{th:2T vs 2tt} and the fact
that every 2-$\T$ reduction is a 3-$\tt$ reduction.
\end{proof}

\begin{corollary}\label{co:3tt vs 1T}
If $\NP$ contains a \pgeneric language, then there exists a
$3$-$\tt$-complete set for $\NP$ that is $3$-$\tt$-autoreducible, but not $1$-$\T$-autoreducible. 
\end{corollary}

Our next theorem separates $(k+1)$-tt-autoreducibility from
$k$-tt-autoreducibility and $k$-T-autoreducibility from
$k$-tt-autoreducibility under the Genericity Hypothesis. The proof
uses the construction of Ambos-Spies and Bentzien \cite{AmbBen00} that
separates the corresponding completeness notions.

\begin{theorem}\label{th:k-tt}
If $\NP$ contains a \pgeneric language, then for every $k \geq 3$
there exists a set that is
\begin{itemize}
\item $(k+1)$-$\tt$-complete for $\NP$ and $(k+1)$-$\tt$-autoreducible,
\item $k$-$\T$-complete for $\NP$ and $k$-$\T$-autoreducible, and
\item not $k$-$\tt$-autoreducible.
\end{itemize}
\end{theorem}
\begin{proof}
Let $G \in \NP$ be a \pgeneric language, and $z_1,...,z_{(k+1)}$ be the first $k+1$ strings of length $k$.
For $m = 1,...,k-1$ define
\begin{align}
  &  \hat{G_m}  = \{x \;|\; xz_m \in G \} \\
  & \hat{G}  = \bigcup\limits_{m=1}^{k-1} \hat{G_m} \\
  &  A  = \bigcup\limits_{m=1}^{k-1}\{xz_m \;|\; x \in \hat{G_m} \} 
          \bigcup \{xz_k \;|\; x \in \hat{G} \cap \SAT \}
          \bigcup \{xz_{k+1}\;|\; x \in \hat{G} \cup \SAT\}
\end{align}
Here are some properties of the sets defined above:
\begin{itemize}

\item For every $x$, $x \in \hat{G} \Leftrightarrow \exists 1 \le i \le k-1.\; xz_i \in G$. 
\item $A$ contains strings in $G$ that end with $z_1,...$, or $z_{(k-1)}$, i.e. $A(xz_i) = G(xz_i)$ for every $x$ and $1 \le i \le k-1$.
\item $xz_k \in A $ if and only if $x \in \SAT \wedge (\exists 1 \le i \le k-1. xz_i \in G)$.
\item $xz_{(k+1)} \in A $ if and only if $x \in \SAT \vee (\exists 1 \le i \le k-1. xz_i \in G)$.
\item $xz_j \notin A$ for $j > k+1$.
\end{itemize}
It is easy to show that $\SAT \le_{(k+1)-\tt}^\p A$. On input $x$, make
queries $xz_1,...,xz_{(k+1)}$ from $A$.  If at least one of the
answers to the first $k-1$ queries is positive, then $\SAT[x]$ is
equal to the $k$th query, i.e. $\SAT[x] = A[xz_k]$. Otherwise
$\SAT[x]$ is equal to $A[xz_{(k+1)}]$. As a result, $A$ is
$(k+1)$-$\tt$-complete for $\NP$. If the queries are allowed to be
dependent, we can choose between $xz_k$ and $xz_{(k+1)}$ based on the
answers to the first $(k-1)$ queries. Therefore $A$ is also
$k$-$\T$-complete for $\NP$. Since all these queries are honest, in
fact length-increasing, it follows from Lemma \ref{le:honest auto} that $A$
is both $(k+1)$-$\tt$-autoreducible and $k$-$\T$-autoreducible.

To get a contradiction, assume $A$ is $k$-$\tt$-autoreducible via $h,
g_1, ...,g_k$. In other words, assume that for every $x$:
\begin{equation}
A[x] = h(x, A[g_1(x)],...,A[g_k(x)])
\end{equation}
and $\forall 1 \le i \le k.\; g_i(x) \neq x$.
In particular, we are interested in the case where $x = 0^nz_1 = 0^{n+k}$, and we have:
\begin{equation}\label{auto-eq}
A(0^{n+k}) = h(0^{n+k}, A[g_1(0^{n+k})],...,A[g_k(0^{n+k})])
\end{equation}
and all $g_i(0^{n+k})$'s are different from $0^{n+k}$ itself.

In the following we will define a bounded extension function $f$ that satisfies the condition in Lemma 
\ref{extension function revised} such that if $G$ meets $f$ at $0^{n+k}$ then \eqref{auto-eq} will fail.
We use the \pgenericity of $G$ to show that $G$ has to meet $f$ at $0^{n+k}$ for some $n$ which completes the proof. In other words, we define a bounded extension function $f$ such that 
given $n$ and $X \upharpoonright 0^n$, $f(X \upharpoonright 0^n) = (y_0,i_0)...(y_m,i_m)$ and if
\begin{equation}\label{meeting conditions}
\begin{aligned}
& G \upharpoonright 0^n = X \upharpoonright 0^n \;\;\textrm{and}\\
&\forall 0 \le j \le m.\;G(y_j)=i_j
\end{aligned}
\end{equation}
then
\begin{equation}\label{refuting auto-red}
A(0^{n+k}) \neq h(0^{n+k}, A[g_1(0^{n+k})],...,A[g_k(0^{n+k})])
\end{equation}
Moreover, $m$ is bounded by some constant that does not depend on $n$ and $X \upharpoonright 0^n$.
Note that we want $f$ to satisfy the conditions in Lemma \ref{extension function revised}, so 
$y_j$'s and $i_j$'s must be computable in $O(2^n)$ and $O(2^{|y_j|})$ steps respectively. After defining such $f$, by Lemma \ref{extension function revised} $G$ must meet $f$ at $0^{n+k}$ for some $n$. This means
\eqref{meeting conditions} must hold. As a result, \eqref{refuting auto-red} must happen for some $n$, 
which is a contradiction.\\
$f$ can force values of $G[y_i]$'s for a constant number of $y_i$'s. Because of the dependency between $G$ and $A$ we can force values for $A[w]$, where $w$ is a query, by using $f$ to force values in $G$. This is done based on the strings that have been queried, and their indices as follows. 
\begin{itemize}
\item If $w = vz_i$ for some $1 \le i \le k-1$ then $A[w] = G[w]$. Therefore we can force $A[w]$ to $0$ or $1$ by forcing the same value for $G[w]$.
\item If $w = vz_k$ then $A[w] = \SAT[v] \wedge (\bigvee_{l=1}^{k-1}{G[vz_l]})$, so by forcing all
$G[vz_l]$'s to $0$ we can make $A[w] = 0$.
\item If $w = vz_{k+1}$ then $A[w] = \SAT[v] \vee (\bigvee_{l=1}^{k-1}{G[vz_l]})$. In this case by forcing one of the $G[vz_l]$'s to $1$ we can make $A[w] = 1$.
\end{itemize}
We will use these facts to force the value of $A$ on queries on input $0^{n+k}$ on the left hand side of
\eqref{auto-eq}, and then force a value for $A[0^{n+k}]$ such that \eqref{auto-eq} fails. The first problem that we encounter is the case where we have both $vz_k$ and $vz_{k+1}$ among our queries. If this happens
for some $v$ then the strategy described above does not work. To force $A[vz_k]$ and $A[vz_{k+1}]$ to 
$0$ and $1$ respectively, we need to compute $\SAT[v]$. If $\SAT[v] = 0$ then $A[vz_k] = 0$, 
and $A[vz_{k+1}]$ can be forced to $1$ by forcing $G[vz_l] = 1$ for some $1 \le l \le k-1$. 
On the other hand, if $\SAT[v] = 1$ then $A[vz_{k+1}] = 1$, and forcing all $G[vz_l]$'s to $0$ makes 
$A[vz_k] = 0$. This process depends on the value of $\SAT[v]$, and $v$ can be much longer that $0^{n+k}$. 
Because of the time bounds in Lemma \ref{extension function revised} the value forced for $A[0^{n+k}]$ cannot depend on $\SAT[v]$. But note that we have $k$ queries, and two of them are $vz_k$ and $vz_{k+1}$. Therefore at least one of the strings $vz_1,...,vz_{k-1}$ is not among the queries. We use this 
string as $vz_l$, and make $G[vz_l] = 1$ when $\SAT[v] = 0$.

Now we define an auxiliary function $\alpha$ from the set of queries, called $\QUERY$, to $0$ or $1$. The idea is that 
$\alpha$ computes the value of $A$ on queries without computing $G[v]$, given that $G$ meets the extension function. $\alpha$ is defined in two parts based on the length of the queries.
For queries $w = vz_p$ that are shorter than $0^{n+k}$, i.e. $|w| < n+k$, we define:
\[ \alpha(w) = \begin{cases} 
      X[w] & \textrm{if $1 \le p \le k-1$} \\
      1 & \textrm{if $p = k \;\wedge \; v \in \SAT \;\wedge \; \exists 1 \le l \le k-1. \; vz_l \in X$} \\
      1 & \textrm{if $p= k+1 \;\wedge \;( v \in \SAT \;\vee \; \exists 1 \le l \le k-1. \; vz_l \in X)$} \\  
      0 & \textrm{otherwise}
   \end{cases} \]
This means that if $X \upharpoonright 0^{n+k} = G \upharpoonright 0^{n+k}$ then $\alpha (w) = A(w)$ for
every query $w = vz_p$ with $|w| < n+k$.\\
On the other hand, for queries $w = vz_p$ that $|w| \geq n+k$, $\alpha$ is defined as:
   \[ \alpha(w) = \begin{cases} 
      1 & \textrm{ if $v = 0^n \;\wedge  p = 2$} \\
      \SAT[v] & \textrm{ if $v = 0^n \;\wedge \; p = k$} \\
      1 &  \textrm{if $v = 0^n \;\wedge \; p = k+1$} \\  
      1 & \textrm{if $ v \neq 0^n \wedge p = k+1$}\\
      1 & \textrm{if $v \neq 0^n \wedge p = k-1 \wedge 
      \forall l \in \{1,...,k-1,k+1\}.\; vz_l \in \QUERY$}\\
      0 & \textrm{otherwise}\\
   \end{cases} \]
For this part of $\alpha$, our definition of the extension function, which is provided below, guarantees that $\alpha(w) = A[w]$ if \eqref{meeting conditions} holds. Note that the first case in the definition above implies that $k$ must be greater than or equal to $3$, and that is the reason this proof does not work for separating 
$3$-$\tt$-autoreducibility from $2$-$\tt$-autoreducibility.

Now we are ready to define the extension function $f$. For any string $v$ which is the value for 
some query, i.e. $\exists 1 \le p \le k+1. vz_p \in \QUERY$, we define pairs of strings and $0$ or $1$'s. These pairs will be part of our extension function. Fix some value $v$, and let $r$ be the smallest index 
that $vz_r \notin \QUERY$, or $k-1$ if such index does not exist, i.e.
\begin{equation}
r = min \{s \geq 1 | vz_s \notin \QUERY \vee s = k-1 \}
\end{equation}
We will have one of the following cases:
\begin{enumerate}
\item If $v = 0^n$ then pairs $(vz_2,1),(vz_3,0),...,(vz_{k-1},0)$ must be added to $f$.
\item If $v \neq 0^n$ and $vz_{k+1} \notin \QUERY$ then add pairs 
$(vz_1,0)$,...,$(vz_{k-1},0)$ to $f$.
\item If $v \neq 0^n$, $vz_{k+1} \in \QUERY$ and $vz_k \notin \QUERY$ add pairs $(vz_i,j)$ for 
$1 \le i \le k-1$ where $j = 0$ for all $i$'s except $i = r$ where $j = 1$.
\item If $v \neq 0^n$, $vz_{k+1} \in \QUERY$ and $vz_k \in \QUERY$ add pairs 
$(vz_i,j)$ for $1 \le i \le k-1$ where $j = 0$ for all $i$'s except $i = r$ where $j = 1 - \SAT[v]$.
\end{enumerate}
This process must be repeated for every $v$ that is the value of some query. Finally, we add 
 $(0^{n+k}, 1 - h(0^{n+k}, \alpha(g_1(0^{n+k})),...,\alpha(g_k(0^{n+k})))$ to $f$ in order to
 refute the autoreduction.
It is worth mentioning that in the fourth case above, since both $vz_k$ and $vz_{k+1}$ are among queries,
at least one of the strings $vz_1$,...,$vz_{k-1}$is not queried. Therefore by definition of $r$, 
$vz_r \notin \QUERY$. This is important, as we describe in more detail later, because we forced
$G[vz_r] = 1 - \SAT[v]$, and if $vz_r \in \QUERY$ then $\alpha(vz_r) = G[vz_r] = 1 - \SAT[v]$. 
But $\alpha$ must be compuatable in $O(2^n)$ steps, which is not possible if $v$ is much longer
than $0^{n+k}$.

Now that the extension function is defined completely, we need to show that it has the desired properties.
First, we will show that if $G$ meets $f$ at $0^{n+k}$, i.e. \eqref{meeting conditions} holds, then 
$\alpha$ and $A$ agree on every query $w$ with $|w| \geq n+k$, i.e. $\alpha(w) = A[w]$.\\
Let $w = vz_p$, and $|w| \geq n+k$.
\begin{itemize}
\item If $v = 0^n$ and $p = 2$ then $\alpha(w) = 1$ and $A[w] = G[w] = 1$.
\item If $v = 0^n$ and $p = k$ then $\alpha(w) = \SAT[v]$ and 
$A[w] = \SAT[v] \wedge (\bigvee_{l = 1}^{k-1}{G[vz_l]})$. Since $G[vz_2] = 1$ is forced, $A[w] = \SAT[v]$.
\item If $v = 0^n$ and $p = k+1$ then $\alpha(w) = 1$ and $A[w] = \SAT[v] \vee (\bigvee_{l = 1}^{k-1}{G[vz_l]}) = 1$ since $G[vz_2] = 1$.
\item If $v = 0^n$ and $p \neq 2,k,k+1$ then $\alpha(w) = A[w] = 0$.
\item If $v \neq 0^n$ and $ p < k-1$ then $\alpha(w) = 0$. Since $p < k-1$, and $vz_p \in \QUERY$, 
by definition of $r$, $r \neq p$. Therefore $G[vz_p]$ is forced to $0$ by $f$. As a result,
 $A[w] = A[vz_p] = G[vz_p] = 0 = \alpha (w)$.
\item If $v \neq 0^n$, $p = k-1$, and $vz_1$,...,$vz_{k-1}$,$vz_{k+1} \in \QUERY$ then
$\alpha (w) = 1$. In this case $r = k-1$, so it follows from definition of $f$ that $G[vz_{k-1}] = 1$.
As a result, $A[w] = A[vz_{k-1}] = G[vz_{k-1}] = 1 = \alpha(w)$.
\item If $v \neq 0^n$, $p = k-1$, and at least one of the strings $vz_1$,...,$vz_{k-1}$,$vz_{k+1}$ 
is not queried then we consider two cases. If $vz_{k+1} \notin \QUERY$ then $f$ forces $G[vz_{k-1}]$ to
$0$. On the other hand, if $vz_{k=1} \in \QUERY$, then at least one of $vz_1$,...,$vz_{k-1}$ 
is not a query. Therefore by definition of $r$, $r \neq k-1$. This implies that $G[vz_{k-1}] = 0$ by $f$.
\item If $v \neq 0^n$, $p = k$ then $\alpha (w) = 0$. Consider two cases. If $vz_{k+1} \notin \QUERY$
then $G[vz_i] = 0$ for every $1 \le i \le k-1$. 
Therefore $A[w] = \SAT[v] \wedge (\bigvee_{l = 1}^{k-1}{G[vz_l]}) = 0$.
Otherwise, when $vz_{k+1} \in \QUERY$, since we know that $vz_k$ also belongs to $\QUERY$, $f$ forces $G[vz_r] = 1 - \SAT[v]$, and $G[vz_i] = 0$ for every other $1 \le i \le k-1$. 
Therefore $A[w] = \SAT[v] \wedge (\bigvee_{l = 1}^{k-1}{G[vz_l]}) = \SAT[v] \wedge (1-\SAT[v]) = 0$.
\item If $v \neq 0^n$, $p = k + 1$ then $\alpha (w) = 1$. If $vz_k \notin \QUERY$
then $G[vz_r] = 1$ by $f$. Therefore $A[w] = \SAT[v] \vee (\bigvee_{l = 1}^{k-1}{G[vz_l]}) = 1$.
On the other hand, if $vz_k \in \QUERY$ then $f$ forces $G[vz_r] = 1 - \SAT[v]$. As a result, 
$A[w] = \SAT[v] \vee (\bigvee_{l = 1}^{k-1}{G[vz_l]}) = 1$.
\end{itemize}
This shows that in any case, $\alpha (w) = A[w]$ for $w \in \QUERY$, given that \eqref{meeting conditions}
holds, i.e $G$ meets $f$. By combining this with \eqref{auto-eq} we have
\begin{equation*}
\begin{aligned}
A(0^{n+k}) = & h(0^{n+k}),A(g_1(0^{n+k})),...,A(g_k(0^{n+k}))) \\
=& h(0^{n+k},\alpha (g_1(0^{n+k})),...,\alpha(g_k(0^{n+k})))
\end{aligned}
\end{equation*}
On the other hand, we forced 
$A[0^{n+k}] = 1 - h(0^{n+k},\alpha (g_1(0^{n+k})),...,\alpha(g_k(0^{n+k})))$
which gives us the desired contradiction.

The last part of our proof is to show that $f$ satisfies the
conditions in Lemma \ref{extension function revised}.  For every value $v$
which is the value of some query we added $k-1$ pairs to $f$, and
there are $k$ queries, which means at most $k$ different
values. Therefore, the number of pairs in $f$ is bounded by $k^2$,
i.e. $f$ is a bounded extension function.

If $f(X \upharpoonright
0^{n+k}) = (y_0,j_0),...,(y_m,j_m)$ then $y_i$'s are computable in
polynomial ime in $n$, and $j_i$'s are computable in $O(2^{|y_i|})$
because the most time consuming situation is when we need to compute
$\SAT[v]$ which is doable in $O(2^n)$. For the condition forced to the
left hand side of \eqref{auto-eq}, i.e $G[0^{n+k}] = 1 -
h(0^{n+k},\alpha (g_1(0^{n+k})),...,\alpha(g_k(0^{n+k})))$, note that
$\alpha (w)$ can be computed in at most $O(2^n)$ steps for $w \in
\QUERY$, and $h$ is computable in polynomial time.
\end{proof}

Next we separate $(k+1)$-$\tt$-autoreducibility and
$k$-$\T$-autoreducibility from $(k-1)$-$\T$-autoreducibility. The
proof uses the same construction from the previous theorem, which
Ambos-Spies and Bentzien \cite{AmbBen00} showed separates these
completeness notions.
\begin{theorem}\label{th:k-1-T}
If $\NP$ contains a \pgeneric language, then for every $k \geq 3$
there exists a set that is
\begin{itemize}
\item $(k+1)$-$\tt$-complete for $\NP$ and $(k+1)$-$\tt$-autoreducible,
\item $k$-$\T$-complete for $\NP$ and $k$-$\T$-autoreducible, and
\item not $(k-1)$-$\T$-autoreducible.
\end{itemize}
\end{theorem}
\begin{proof}
We use the same sets $G$ and $A$ as defined in the proof of Theorem
\ref{th:k-tt}. We proved that $A$ is $(k+1)$-$\tt$-complete,
$k$-$\T$-complete, $(k+1)$-$\tt$-autoreducible, and
$k$-$\T$-autoreducible.  What remains is to show that it is not
$(k-1)$-$\T$-autoreducible. The proof is very similar to what we did in
the previous theorem, so we will not go through every detail
here. Assume $A$ is $k$-$\T$-autoreducible via an oracle Turing machine
$M$. In other words,
\begin{equation}\label{T-auto-red-eq}
\forall x.\;A[x] = M^A(x)
\end{equation}
and we assume that on input $x$, $M$ will not query $x$ itself. By using \pgenericity of $G$ we will show that there exists some $n$ such that \ref{T-auto-red-eq} fails for $x = 0^{n+k}$. In other words,
\begin{equation}
\exists n.\;A[0^{n+k}] = M^A(0^{n+k})
\end{equation}
Similar to what we did in the previous theorem, we define a bounded extension function $f$
such that given $n$ and an initial segment $X\upharpoonright 0^n$, $f$ returns a set of
pairs $(y_i,j_i)$ for $0 \le i \le m$. $y_i$'s are the positions, and must be computable 
in $O(2^n)$ steps, and $j_i$'s are the values that $f$ forces to $y_i$'s. Each $j_i$ must be 
computable in $O(2^{|y_i|})$. Then we will show that if $G$ meets $f$ at $0^{n+k}$, i.e. if 
\ref{meeting conditions} holds, then \ref{T-auto-red-eq} fails for $x = 0^{n+k}$. We will define a function $\alpha$ that under the right conditions simulates $A$ on queries. We use $\alpha$ 
instead of $A$, as the oracle, in the computation of $M$ on input $0^{n+k}$. Similar to the previous theorem, $\alpha$ must be computable in $O(2^n)$ steps. Since in a Turing reduction each query may depend on the answers to the previous queries, we cannot know which queries will be asked in the computation 
of $M^A(0^{n+k})$ in $O(2^n)$ steps. Therefore we define $\alpha$ on every string rather than 
just on the set of queries.

Let $w = vz_p$ be some string. If $|w| < n+k$, then $\alpha$ is defined as:
\[ \alpha(w) = \begin{cases} 
      X[w] & \textrm{ if $1 \le p \le k-1$} \\
      1 & \textrm{ if $p = k \;\wedge \; v \in \SAT \;\wedge \; \exists 1 \le l \le k-1. \; vz_l \in X$} \\
      1 & \textrm{if $p= k+1 \;\wedge \;( v \in \SAT \;\vee \; \exists 1 \le l \le k-1. \; vz_l \in X)$} \\  
      0 & \textrm{otherwise}\\
   \end{cases} \]
and if $|w| \geq n+k$ then:
\[ \alpha(w) = \begin{cases} 
      1 & \textrm{ if $v = 0^n \;\wedge  p = 2$} \\
     \SAT[v] & \textrm{ if $v = 0^n \;\wedge \; p = k$} \\
     1 &  \textrm{if $p = k+1$} \\  
     0 & \textrm{otherwise}\\
\end{cases} \]
Now we run the same oracle Turing machine $M$, but we use $\alpha$ as the oracle instead of $A$. Let
$\QUERY$ be the set of queries asked in this process. $f$ will be defined in a similar fashion, except
that the final pair which completes the diagonalization would be $(0^{n+k}, 1 - M^{\alpha}(0^{n+k}))$.
Note that because there are at most $k-1$ queries in both cases $3$ and $4$ in the definition of $f$,
$vz_r \notin \QUERY$. In other words, the string we are forcing into $G$ (hence into $A$) will never be queried.

Similar to the previous theorem, it can be verified that $\alpha$ and
$A$ agree on all queries, i.e.  $M^A(0^{n+k}) = M^{\alpha}(0^{n+k})$,
if \ref{meeting conditions} holds. It is also easy to prove that
$\alpha$ is computable in $O(2^n)$ steps, therefore $f$ satisfies the
time bounds in Lemma \ref{extension function revised}.
\end{proof}

We now separate unbounded truth-table autoreducibility from bounded
truth-table 
\\ autoreducibility under the Genericity Hypothesis. This is
based on the technique of Ambos-Spies and Bentzien \cite{AmbBen00}
separating the corresponding completeness notions.

\begin{theorem}\label{th:btt vs tt}
If $\NP$ has a \pgeneric language, then there exists a $\tt$-complete set for $\NP$ that is 
$\tt$-autoreducible, but not $\btt$-autoreducible.
\end{theorem}
Before proving Theorem \ref{th:btt vs tt}, we need a few definitions
and two lemmas.

A complexity class $C$ is {\em computably presentable} if there is a
computable function $f : \N \to \N$ such that $C = \{ L(M_{f(i)}) \mid
i \in \N \}$. A sequence of classes $C_0,C_1,\ldots$ is {\em uniformly
  computably presentable} if there is a computable function $f : \N
\times \N \to \N$ such that $C_j = \{ L(M_{f(j,i)}) \mid i \in \N \}$
for all $j \in \N$. A reducibility $\calR$ is {\em computably
  presentable} if there is a computable function $f : \N \to \N$ such
that $M_{f(1)},M_{f(2)},\ldots$ is an enumeration of all $\calR$-reductions.


\begin{lemma}\label{enumeration}
If $C$ is a computably presentable class which is closed under finite
variants and $\calR$ is a computably presentable reducibility, then 
$C_{\calR\text{-}auto} = \{B \in C \;|\; B \; is \; \calR$-$autoreducible \}$ is also computably presentable.
\end{lemma}
\begin{proof}
We prove the lemma for polynomial-time Turing autoreducibility, but similar proofs can be constructed for 
any kind of autoreduction that is computably presentable. For simplicity, we use $C_{auto}$ for 
$C_{poly \text{-}T\text{-}auto}$ in the rest of the proof. If
$C_{auto} = \emptyset$ then it is computably presentable by
convention. Assume $C_{auto} \neq \emptyset$, and fix some set $A \in
C_{auto}$. Since $C$ is closed under finite variants, any finite
variation of $A$ must also belong to $C_{auto}$.

Let $N_1$,$N_2$,... be a presentation of $C$, and $T_1$,$T_2$,... be an enumeration of deterministic polynomial-time oracle Turing machines. For every pair $n = \pair{i,j}$ where $i,j \geq 1$ we define a Turing machine $M_n$ as follows:\\
\begin{tabular}{lll}
\toprule
$M_n$ &    &\\
\midrule
input $x$ &     & \\
for each y with $y < x$ do & &\\
& test that $y \in L[N_i] \Leftrightarrow y \in L(T_j,L(N_i))$, & \\
& and $y$ itself has not been queried by $T_j$ &\\
if tests are true & &\\
& then accept $x$ iff $x \in L(N_i)$ & \\
& else accept $x$ iff $x \in A$ & 
\end{tabular}\\

Let $L$ be an arbitrary language in $C_{auto}$. There must be some $i,j \geq 1$ such that $L = L(N_i)$ 
and $T_j$ computes an $\calR$-autoreduction on $L$. Therefore $M_n$ computes $L$ when $n = \pair{i,j}$.
This means that every language in $C_{auto}$ is accepted by some
Turing machine $M_n$. On the other hand, for every $n=\pair{i,j}$, 
 if $T_j$ does not compute an $\calR$-autoreduction on $L(N_i)$, then
 $L(M_n)$ is a finite variant of $A$. Since $C$ is assumed to be closed under finite variants, $L(M_n) \in C_{auto}$. 
\end{proof}
\begin{lemma_cite}{Ambos-Spies and Bentzien \cite{AmbBen00}}\label{diagonalization-lemma}
Let $C_0,C_1,\ldots$ be classes such that,\\
$(1)$. $C_0,C_1,\ldots$ is uniformly computably presentable.\\
$(2)$. Each $C_i$ is closed under finite variants.\\
$(3)$. There is a decidable set $D$ such that $D \subseteq \{0\}^* \times \Sigma ^*$ , \\
and $D^{[n]} = \{x | <0^n,x> \in D \} \notin C_n$.\\
$(4)$. $f : N \rightarrow N$ is a non-decreasing unbounded computable function.\\
Then there exists a set $A$ and a function $g : N \rightarrow N$ such that:\\
$(5)$. $A \notin \bigcup_{n = 0}^{\infty}{C_n}$.\\
$(6)$. $\forall n.\; A_{=n} = D_{=n}^{[g(n)]}$.\\
$(7)$. $g$ is polynomial-time computable with respect to the unary representation of numbers.\\
$(8)$. $\forall n.\; g(n) \leq f(n)$.
\end{lemma_cite}


\begin{proof}[Proof of Theorem \ref{th:btt vs tt}]
Let $\widetilde {\SAT} = \{ 0^n1x \;|\; n \geq 0 \;\textrm{ and }\; x \in \SAT \} $. It is easy to see that 
$\widetilde {\SAT}$ is $\NP$-complete, and $\widetilde{\SAT} \in \DTIME(2^n)$. For every $k \geq 0$, let $A_k$ 
be a $(k+3)$-$\tt$-complete set constructed as before, by using $\widetilde {\SAT}$ instead of $\SAT$, 
and fix a \pgeneric set $G \in \NP$ for the rest of the proof. Note that $A_k$ is also 
$(k+3)$-$\tt$-autoreducible, but not $(k+2)$-$\tt$-complete or $(k+2)$-$\tt$-autoreducible. 
Define $D = \{<0^k,x> \; | \; k \geq 0 \; and \; x \in A_k \}$. Since $A_k \in \NP$ uniformly in $k$, 
$D \in \NP$. Let $C_k = \{ B \in \NP \; | \; B \; \text{is} \; k\text{-}\tt\text{-autoreducible} \}$ for
$k \geq 1$ and $C_0 = C_1$. $\NP$ is computably presentable and
closed under finite variants, therefore by Lemma \ref{enumeration},
$C_k$'s are computably presentable. In fact, they are uniformly
computably presentable by applying the proof of Lemma
\ref{enumeration} uniformly. It is also easy to see that each $C_k$ is closed under finite variants. Therefore $C_k$'s satisfy the conditions of Lemma \ref{diagonalization-lemma}. It follows from the definition of $D$ that $D^{[k]} = A_k$, and we know that $A_k \notin C_k$ by construction of $A_k$. Therefore, if we take $f(n) = min \{m\;|\; 2m+3 \geq n\}$, by Lemma \ref{diagonalization-lemma} there exist $A$ and $g$ such that properties $(5)$-$(8)$ 
from the lemma hold.\\
It follows from $(6)$ and $(7)$ that $\forall n.\; A_{=n} = D_{=n}^{[g(n)]}$, and $g$ is polynomial time 
computable with respect to unary representation of numbers. This implies that $A \le_\m^\p D$, therefore 
$A \in \NP$. Moreover, by $(5)$ from the lemma, $A \notin \bigcup_{n \geq 0} {C_n}$, which means for 
every $k \geq 1$, $A$ is not $k$-$\tt$-autoreducible. In other words $A$ is not $\btt$-autoreducible.

To show that $A$ is $\tt$-autoreducible, we will show that $\SAT \le_{\tt}^\p A$ via honest reductions, and then it follows from Lemma \ref{le:honest auto} that $A$ is $\tt$-autoreducible. To define the truth-table reduction from $\SAT$ to $A$, fix $x$ with $|x| = n$. For every $k,m \geq 0$ we have $\SAT[x] = \widetilde {\SAT}[0^m1x]$,
 and $\widetilde {\SAT}[0^m1x]$ can be computed by making $(k+3)$ independent queries from 
 $(A_k)_{=m+1+n+k+2}$ in polynomial time, uniformly in $x$, $k$, and $m$(This follows from 
 $(k+3)$-$\tt$-completeness of $A_k$, and the way $A_k$ is defined using $\widetilde {\SAT}$. 
 $(7)$ from Lemma \ref{diagonalization-lemma} implies that:
 \begin{equation}
 A_{=2n+3} = (D^{[g(2n+3)]})_{=2n+3} = (A_{g(2n+3)})_{=2n+3}
 \end{equation}
We also know that $g(2n+3) \le f(2n+3) \le n$ for all $n$. Using all these facts, here is the truth-table 
reduction from $\SAT$ to $A$:\\
For $x$ with $|x| = n$, compute $g(2n+3)$, and let $k = g(2n+3)$ and $m = n - k$. Therefore:
\begin{equation}
(A_k)_{=m+1+n+k+2} = (A_{g(2n+3)})_{=2n+3} = A_{=2n+3}
\end{equation}
We know that $\SAT[x] = \widetilde {\SAT}[0^m1x]$ can be computed by making $(k+3)$ independent queries from 
$(A_k)_{=m+1+n+k+2}$. This means $\SAT[x] = \widetilde {\SAT}[0^m1x]$ can be recovered by making $g(2n+3)$ 
queries from $A_{=2n+3}$.\\
Note that all these queries are longer than $x$. Therefore, by Lemma
\ref{le:honest auto}, $A$ is $\tt$-autoreducible. 
\end{proof}
\section{Stronger Separations Under a Stronger Hypothesis}\label{sec:stronger}
Our results so far only separate $k$-tt-autoreducibility from
$(k-2)$-T-autoreducibility for $k \geq 3$ under the genericity
hypothesis. In this section we show that a stronger hypothesis
separates $k$-tt-autoreducibility from $(k-1)$-T-autoreducibility, for
all $k \geq 2$. We note that separating $k$ nonadaptive queries from
$k-1$ adaptive queries is an optimal separation of bounded query
reducibilities.

First we consider $2$-$\tt$-autoreducibility versus
$1$-$\tt$-autoreducibility (equivalently,
$1$-$\T$-autoreducibility). Pavan and Selman \cite{PavSel04} showed
that if $\NP \cap \coNP$ contains a $\DTIME(2^{n^\epsilon})$-bi-immune set, then
$2$-$\tt$-completeness is different from $1$-$\tt$-completeness for
$\NP$. We show under the stronger hypothesis that $\NP \cap \coNP$
contains a \pgeneric set, we can separate the autoreducibility
notions.

\begin{theorem}\label{th:2tt vs 1tt}
If $\NP \cap \coNP$ has a \pgeneric language, then there exists a $2$-$\tt$-complete set for $\NP$ that is 
$2$-$\tt$-autoreducible, but neither $1$-$\tt$-complete nor $1$-$\tt$-autoreducible.
\end{theorem}
\begin{proof}
Assume $G \in \NP \cap \coNP$ is \pgeneric, and let $A = (G \cap \SAT) \dot\cup (\overline{G} \cap \SAT)$, where 
$\overline{G}$ is $G$'s complement, and $\dot\cup$ stands for disjoint union. We implement  disjoint union as 
 $A = (G \cap \SAT)0 \; \dot\cup \;(\overline{G} \cap \SAT)1$. It follows from closure properties of $\NP$ and the fact that $G \in \NP \cap \coNP$ that $A \in \NP$. It follows from definition of $A$ that for every $x$,
 $x \in \SAT \leftrightarrow (x0 \in A \vee x1 \in A)$. This means $\SAT \le_{2\tt}^\p A$. Therefore
  $A$ is $2$-$\tt$-complete for $\NP$. Since both queries in the above reduction are honest, in fact length increasing, it follows from Lemma \ref{le:honest auto} that $A$ is $2$-$\tt$-autoreducible. To get a contradiction assume that $A$ is $1$-$\tt$-autoreducible via polynomial-time computable functions $h$ and $g$. In other words,
\begin{equation}\label{1tt-autored}
 \forall x.\; A(x) = h(x , A[g(x)])
\end{equation}  
and $g(x) \neq x$. Let $x = y0$ for some string $y$, then \eqref{1tt-autored} turns into
\begin{equation}
 \forall y.\; G \cap \SAT[y] = h(y0 , A[g(y0)])
\end{equation}
and $g(y0) \neq y0$.
We define a bounded extension function $f$ whenever $\SAT[y] = 1$ as follows.
\begin{itemize}
  \item Consider the case where $g(y0) = z0$ or $z1$ and
$z > y$.
If $g(y0) = z0$ then $f$ forces
$G[z] = 0$, and if $g(y0) = z1$ then $f$ forces $G[z] = 1$. $f$ also
forces $G[y] = 1 - h(y0 , 0)$. Since $g$ and $h$ are computable in
polynomial time, so is $f$.
\item On the other hand, if $g(y0) = z0$ or $z1$ and $z < y$ then
  define $f$ such that it forces $G[y] = 1 - h(y0, A[g(y0)])$.  Then
  $f$ polynomial-time computable in this case as well because $A$ may
  be computed on $g(y0)$ by looking up $G[z]$ from the partial
  characteristic sequence and deciding $\SAT[z]$ in $2^{O(|z|)}$ time.
\item If $g(y0) = y1$ and $h(y0,.) = c$ is a constant function, 
  then define $f$ such that it forces $G[y] = 1 - c$.   
\end{itemize}
If $g(y0) \neq y1 \wedge \SAT[y] = 1$ for infinitely many $y$, it follows
from the \pgenericity of $G$ that $G$ has to meet $f$, but this refutes the autoreduction.
Similarly, $g(y0) = y1 \wedge h(y0,.) = const \wedge \SAT[y] = 1$ cannot happen for 
infinitely many $y$'s.
As a result, $(g(y0) = y1 \vee \SAT[y] = 0)$ and $h(y0,.)$ is not constant
 for all but finitely many $y$'s. 
If $g(y0) = y1$ then $h$ says either $G \cap \SAT[y] = \overline{G} \cap \SAT [y]$ or 
$G \cap \SAT[y] = \neg ( \overline{G} \cap \SAT [y])$. It is easy to see
this implies $\SAT[y]$ has to be $0$ or $1$, respectively. 
Based on the facts above, we define Algorithm \ref{alg:1tt} that
decides $\SAT$ in polynomial time. This contradicts the assumption
that $\NP \cap \coNP$ has a $\p$-generic language.
\begin{figure}[h!]
\begin{algorithm}[H]
 input y\;
  \eIf{$g(y0) \neq y1 \vee h(y0,.)\textrm{ is constant}$}{
   Output NO\;
   }{
   \eIf{$h(y0,.)\textrm{ is the identity function}$}{
   Output YES\;
   }{
   Output NO\;}
 }
\caption{A polynomial-time algorithm for $\SAT$}\label{alg:1tt}
\end{algorithm}
\end{figure}

It is proved in \cite{Glasser07} that every nontrival
$1$-$\tt$-complete set for $\NP$ is $1$-$\tt$-autoreducible, so it follows that $A$ is not $1$-$\tt$-complete.
\end{proof}

We will show the same hypothesis on $\NP \cap \coNP$ separates
$k$-$\tt$-autoreducibilty from $(k-1)$-$\T$-autoreducibility for all
$k \geq 3$. First, we show the corresponding separation of
completeness notions.

\begin{theorem}\label{th:k-tt vs (k-1)-T complete}
If $\NP \cap \coNP$ contains a \pgeneric set, then for every $k \geq 3$ there exists a $k$-$\tt$-complete 
set for $\NP$ that is not $(k-1)$-$\T$-complete.
\end{theorem}
\begin{proof}
Assume $G \in \NP \cap \coNP$ is \pgeneric, and let $G_m = \{x \;|\; xz_m \in G \}$ 
for $1 \le m \le k$ where $z_1,...,z_k$ are the first $k$ strings of length $k$ as before.
Define 
\begin{equation}
A = \Big[\bigcup _{m = 1}^{k-1} {\{xz_m \;|\; x \in G_m \cap \SAT \}}\Big] \cup 
    \{xz_k \;|\; x \in \big[ \cap _{m = 1}^{k-1}{\overline{G_m}} \big] \cap \SAT \}
\end{equation}
It is easy to check that $x \in \SAT \Leftrightarrow \bigvee _{m =1}^{k} ({xz_m \in A})$, therefore
$\SAT \le _{k-\tt}^{p} A$. It also follows from the fact that $G \in \NP \cap \coNP$ and the closure
properties of $\NP$ that $A \in \NP$, so $A$ is $k$-$\tt$-complete for $\NP$, in fact $k$-$\dtt$-complete.

We claim that $A$ is not $(k-1)$-$\T$-hard for $\NP$. For a contradiction, assume that
$G_k \le_{(k-1)-\T}^{p} A$. In other words, assume that there exists an oracle Turing machine $M$ 
such that 
\begin{equation}
\forall x. \; G_k [X] = M^A[x]
\end{equation}
where $M$ runs in polynomial time, and makes at most $(k-1)$ queries on every input. Given $n$ and 
$X \upharpoonright 0^n$, we define a function $\alpha$ as follows.\\
If $w = vz_p$ and $|w| < n+k$ then
\[ \alpha(w) = \begin{cases} 
      X[w] \wedge \SAT[v] & \textrm{ if $1 \le p \le k-1$} \\
      \big[\bigwedge_{l=1}^{k-1} {(1 - X[vz_l])} \big] \wedge \SAT[v] & 
      \textrm{ if $p = k$} \\
      0 & \textrm{otherwise}\\
   \end{cases} \]
It is easy to see that $\alpha$ is defined in a way that if 
$X \upharpoonright 0^n = G \upharpoonright 0^n$ then $\alpha (w) = A[w]$.\\
On the other hand, if $|w| \geq n+k$ then $\alpha (w) =  0$ all the time. Later when we define the 
extension function we guarantee that $A[w] = 0$ for all long queries, by forcing the right
values into $G$, which implies $A[w] = \alpha (w)$ for all queries. But before doing that, we run
$M$ on input $0^n$ with $\alpha$ as the oracle, and define $\QUERY$ to be the set of all queries
made in this computation. We know that $|\QUERY| \le k-1$ therefore one of the following 
cases must happen:
\begin{enumerate}
\item $xz_k \notin \QUERY$.
\item $xz_k \in \QUERY$, and $\exists 1 \le l \le k-1. \; xz_l \notin \QUERY$.
\end{enumerate}
Define a bounded extension function $f$ based on the above cases. Given $n$ and $X \upharpoonright 0^n$, $f(X \upharpoonright 0^n)$ contains the pairs described below. For every $v$ which is the value of some element of $\QUERY$, 
\begin{enumerate}
\item If $vz_k \notin \QUERY$, then put $(vz_1 , 0)$,...,$(vz_{(k-1)} , 0)$ into $f$. In other words,
$f$ forces $G[vz_l]$ to $0$ for every $1 \le l \le k-1$.
\item If $vz_k \in \QUERY$ then there must be some $1 \le l \le k-1$ such that $vz_l \notin \QUERY$. 
In this case $f$ forces $G[vz_i] = 0$ for every $1 \le i \le k-1$ except for $i = l$ for which $G[vz_l] = 1$.
\end{enumerate}
It can be shown that if $G$ meets $f$ at $0^n$, i.e. if \eqref{meeting conditions} holds, then 
$\alpha (w) = A[w]$ for every $w \in \QUERY$. As a result, 
\begin{equation}\label{last}
M^{\alpha} (0^n) = M^A (0^n)
\end{equation}
To complete the diagonalization, we add one more pair to $f$ that forces the value of $G_k[0^n] = G[0^{n+k}]$ to $1 - M^{\alpha} (0^n)$, i.e. $(0^{n+k} , 1 - M^{\alpha} (0^n))$. Then it follows from \eqref{last}
 that the reduction from $G_k$ to $A$ fails. The last part of the proof, is to show that $G$ has to meet $f$ at $0^n$ for some $n$. $\alpha$ is computable in $O(2^n)$ steps for short queries, and constant time for long queries, and $M$ is a polynomial time Turing machine, which implies $f$ can be computed in 
 at most $O(2^{2n})$ steps. It is also easy to see that the number of
 pairs in $f$ is bounded by $k^2$, which means $f$ is a bounded
 extension function. As a result $f$ satisfies the conditions of Lemma
 \ref{extension function revised}, hence $G$ has to meet $f$ at $0^n$ for some $n$, which completes the proof.
\end{proof}

Now we show the same sets separate $k$-$\tt$-autoreducibility from $(k-1)$-$\T$-autoreducibility.

\begin{theorem}\label{th:k-tt vs (k-1)-T autored}
If $\NP \cap \coNP$ contains a \pgeneric set, then for every $k \geq 3$ there exists a $k$-$\tt$-complete 
set for $\NP$ that is $k$-$\tt$-autoreducible, but is not $(k-1)$-$\T$-autoreducible.
\end{theorem}
\begin{proof}
 Assume $G \in \NP \cap \coNP$ is \pgeneric, and let $G_m = \{x \;|\; xz_m \in G \}$ 
for $1 \le m \le k$ where $z_1,...,z_k$ are the first $k$ strings of length $k$ as before.
Define 
\begin{equation}
A = \Big[\bigcup _{m = 1}^{k-1} {\{xz_m \;|\; x \in G_m \cap \SAT \}}\Big] \cup 
    \{xz_k \;|\; x \in \big[ \cap _{m = 1}^{k-1}{\overline{G_m}} \big] \cap \SAT \}
\end{equation}
We showed that $\SAT \le_{k-\tt}^\p A$ via length-increasing queries, therefore by Lemma \ref{le:honest auto} $A$ is $k$-$\tt$-autoreducible.
For a contradiction, assume that $A$ is $(k-1)$-$\T$-autoreducible. This means there exists an oracle Turing machine $M$ such that 
\begin{equation}
\forall x.\; A[x] = M^A(x)
\end{equation}
$M$ runs in polynomial time, and on every input $x$ makes at most
$k-1$ queries, none of which is $x$. Given $n$ and 
$X \upharpoonright 0^n$, we define a function $\alpha$ as follows.\\
If $w = vz_p$ and $|w| < n+k$ then
\[ \alpha(w) = \begin{cases} 
      X[w] \wedge \SAT[v] & \textrm{ if $1 \le p \le k-1$} \\
      \big[\bigwedge_{l=1}^{k-1} {(1 - X[vz_l])} \big] \wedge \SAT[v] & 
      \textrm{ if $p = k$} \\
      0 & \textrm{otherwise}\\
   \end{cases} \]
It is easy to see that if $X \upharpoonright 0^n = G \upharpoonright 0^n$ then $\alpha (w) = A[w]$.\\
If $w = vz_p$ and $|w| \geq n+k$, $\alpha$ is defined as:
\[ \alpha(w) = \begin{cases} 
      1 & \textrm{ if $v = 0^n \wedge 2 \le p \le k-1$} \\
      0 & \textrm{ if $v = 0^n \wedge p = k$} \\
      0 & \textrm{otherwise}\\
   \end{cases} \]
Note that $\alpha$ is not defined on $0^{n+k}$, but that is fine because we are using $\alpha$ to compute 
$A[w]$ for $w$'s that are queried when the input is $0^{n+k}$, therefore $0^{n+k}$ will not be queried. Later we will define the extension function $f$ in a way that if $G$ meets $f$ at $0^n$ then $\alpha(w) = A[w]$ for all queries.\\
Before defining $f$, we run $M$ on input $0^{n+k}$ with $\alpha$ as the oracle instead of $A$, and define $\QUERY$ 
to be the set of all queries made in this computation. We know that $M$ makes at most $k-1$ queries, therefore 
$|\QUERY| \le k-1$. This implies that for every $v \neq 0^n$ which is the value of some element of $\QUERY$ one of the following cases must happen:
\begin{enumerate}
\item $vz_k \notin \QUERY$
\item $vz_k \in \QUERY$ and $\exists 1 \le l \le k-1 \;.\; vz_l \notin \QUERY$
\end{enumerate}
Given $n$ and $X \upharpoonright 0^n$, $f(X \upharpoonright 0^n)$ is defined as follows if $\SAT[0^n] = 1$.\\
For every $v$ which is the value of some element of $\QUERY$,
\begin{enumerate}
\item If $v = 0^n$, then add $(vz_2,1),...,(vz_{k-1},1)$ to $f$. In other words, $f$ forces $G[0^nz_i] = 1$ for $2 \le i \le k-1$.
\item If $v \neq 0^n$ and $vz_k \notin \QUERY$, then add $(vz_1,0),...,(vz_{k-1},0)$ to $f$.
\item If $v \neq 0^n$ and $vz_k \in \QUERY$, then there must be some $1 \le l \le k-1$ such that $vz_l \notin \QUERY$.
In this case $f$ forces $G[vz_i] = 0$ for every $1 \le i \le k-1$ except when $i = l$ for which we force $G[vz_l] = 1$.
\end{enumerate} 
To complete the diagonalization we add one more pair to $f$ which is $(0^{n+k} , 1 - M^{\alpha}(0^n))$. It is straightforward, and similar to what has been done in the previous theorem, to show that if $G$ meets $f$ at $0^n$ for some $n$ then $\alpha$ and $A$ agree on every element of $\QUERY$. Therefore $M^{\alpha}(0^n) = M^A(0^n)$, which results in a contradiction. 
It only remains to show that $G$ meets $f$ at $0^n$ for some $n$. This
depends on the details of the encoding used for $\SAT$. If $\SAT[0^n] = 1$ for infinitely many $n$'s, then $f$ satisfies the conditions in Lemma \ref{extension function revised}. Therefore $G$ has to meet $f$ at $0^n$ for some $n$. On the other hand, 
if $\SAT[0^n] = 0$ for almost all $n$, then we redefine $A$ as:
\begin{equation}
A = \Big[\bigcup _{m = 1}^{k-1} {\{xz_m \;|\; x \in G_m \cup \SAT \}}\Big] \cup 
    \{xz_k \;|\; x \in \big[ \cup _{m = 1}^{k-1}{\overline{G_m}} \big] \cup \SAT \}
\end{equation}
It can be proved, in a similar way and by using the assumption that $\SAT[0^n] = 0$ for almost all $n$, that $A$ is $k$-$\tt$-complete, $k$-$\tt$-autoreducible, but not $(k-1)$-$\T$-autoreducible. 
\end{proof}

\section{Conclusion}\label{sec:conclusion}

We conclude with a few open questions. 

For some $k$, is there a $k$-tt-complete set for $\NP$ that is not
btt-autoreducible? We know this is true for $\EXP$ \cite{BFvMT00}, so
it may be possible to show under a strong hypothesis on $\NP$.  We
note that by Lemma \ref{le:honest auto} any construction of a
$k$-tt-complete set that is not $k$-tt-autoreducible must not be
honest $k$-tt-complete. In fact, the set must be complete under
reductions that are neither honest nor dishonest. On the other hand,
for any $k \geq 3$, proving that all $k$-tt-complete sets for $\NP$
are btt-autoreducible would separate $\NP \neq \EXP$.

Are the 2-tt-complete sets for $\NP$ 2-tt-autoreducible? The answer to
this question is yes for $\EXP$ \cite{BuhTor05}, so in this case a
negative answer for $\NP$ would imply $\NP \neq \EXP$. We believe that
it may be possible to show the 2-tt-complete sets are nonuniformly
2-tt-autoreducible under the Measure Hypothesis -- first show they are
nonuniformly 2-tt-honest complete as an extension of
\cite{Hitchcock:CRNPCS,BuhrmanHescottHomerTorenvliet:NonuniformReductions}.

Nguyen and Selman \cite{NguyenSelman:STACS14} showed there is
T-complete set for $\NEXP$ that is not tt-autoreducible. Can we do
this for $\NP$ as well? Note that Hitchcock and Pavan
\cite{Hitchcock:CRNPCS} showed there is a T-complete set for $\NP$
that is not tt-complete.

{\bf Acknowledgment.} We thank A. Pavan for extremely helpful discussions.

\bibliographystyle{plain} 
\bibliography{\INC final}

\end{document}